\documentclass[USEnglish]{oasics-v2021}
\nolinenumbers

\usepackage{listings}
\usepackage{url}
\usepackage{hyperref}
\usepackage{todonotes}
\usepackage{menukeys}

\newcommand{\var}[1]{\text{\lstinline+#1+}}

\newtheorem{axiom}{Axiom}

\graphicspath{{./figures/}}

\newcommand{\bv}{\ensuremath{{\boldsymbol v}}}   

\newcommand{\bx}{\ensuremath{{\boldsymbol x}}}

\newcommand{\set}[1]{\left\{ #1 \right\}}

\DeclareMathOperator{\angslip}{angslip}
\DeclareMathOperator{\linslip}{linslip}
\DeclareMathOperator{\divloss}{divloss}
\DeclareMathOperator{\capital}{cap}
\DeclareMathOperator{\numcap}{numcap}
\DeclareMathOperator{\load}{load}
\DeclareMathOperator{\betaDist}{beta}

\newcommand{\figlabel}[1]{\label{figure:#1}}
\newcommand{\nakedfigref}[1]{\ref{figure:#1}}
\newcommand{\figref}[1]{Figure~\nakedfigref{#1}}
\newcommand{\Figref}[1]{Figure~\nakedfigref{#1}}

\newcommand{\eqnlabel}[1]{\label{eq:#1}}
\newcommand{\nakedeqnref}[1]{\ref{eq:#1}}
\newcommand{\eqnref}[1]{Equation~\nakedeqnref{#1}}

\newcommand{\eqnrange}[2]{Equations~\nakedeqnref{#1}-\nakedeqnref{#2}}

\newcommand{\lemmalabel}[1]{\label{lemma:#1}}
\newcommand{\nakedlemmaref}[1]{\ref{lemma:#1}}
\newcommand{\lemmaref}[1]{Lemma~\nakedlemmaref{#1}}

\newcommand{\axiomlabel}[1]{\label{axiom:#1}}

\newcommand{\thmlabel}[1]{\label{thm:#1}}
\newcommand{\nakedthmref}[1]{\ref{thm:#1}}
\newcommand{\thmref}[1]{Theorem~\nakedthmref{#1}}

\newcommand{\seclabel}[1]{\label{sec:#1}}
\newcommand{\nakedsecref}[1]{\ref{sec:#1}}
\newcommand{\secref}[1]{Section~\nakedsecref{#1}}

\newcommand{\Reals}{\ensuremath{\mathbb{R}}}
\newcommand{\PosReals}{\ensuremath{\mathbb{R}_{>0}}}

\title{Presentation and Publication:
Loss and Slippage in Networks of Automated Market Makers}

\author{Daniel Engel}
{Computer Science Department, Brown University, Providence, RI USA}
{daniel\_engel1@brown.edu}
{}
{This research was supported by NSF grant 1917990}

\author{Maurice Herlihy}
{Computer Science Department, Brown University, Providence, RI USA}
{maurice.herlihy@gmail.com}
{https://orcid.org/0000-0002-3059-8926}
{This research was supported by NSF grant 1917990}
\authorrunning{D. Engel and M. Herlihy} 
\begin{document}


\maketitle

\begin{abstract}
Automated market makers (AMMs) are smart contracts that automatically trade electronic assets
according to a mathematical formula.
This paper investigates how an AMM's formula affects the interests
of liquidity providers, who endow the AMM with assets,
and traders, who exchange one asset for another at the AMM's rates.
\emph{Linear slippage} measures how a trade's size affects the trader's return,
\emph{angular slippage} measures how a trade's size affects the subsequent market price,
\emph{divergence loss} measures the opportunity cost of providers' investments,
and \emph{load} balances the costs to traders and providers.
We give formal definitions for these costs,
show that they obey certain conservation laws:
these costs can be shifted around but never fully eliminated.
We analyze how these costs behave under \emph{composition},
when simple individual AMMs are linked to form more complex networks of AMMs.

\end{abstract}

\section{Introduction}
\seclabel{intro}
An \emph{automated market maker} (AMM) is an automaton that trades electronic assets according to a fixed formula.
Unlike traditional ``order book'' traders,
AMMs have custody of their own asset pools,
so they can trade directly with clients,
and do not need to match up (and wait for) compatible buyers and sellers.
Today,
AMMs such as Uniswap~\cite{AngerisKCNC2019}, Bancor~\cite{bancor}, and others have become one
of the most popular ways to trade electronic assets
such as cryptocurrencies, electronic securities, or tokens.
An AMM is typically implemented as a smart contract on a blockchain such as Ethereum~\cite{ethereum}.
Like circuit elements,
AMMs can be \emph{composed} into networks.
They can be composed sequentially,
where the output of one AMM's trade is fed to another,
and they can be composed in parallel,
where a trade is split between two AMMs with different formulas.
Compositions of AMMs can themselves be treated as AMMs~\cite{EngelH2021}.
This paper makes the following contributions.
AMMs have well-known inherent costs.
One such cost is \emph{slippage},
where a large trade increases the price of the asset being purchased,
both for the trader making the trade, and for later traders.
We give two alternative mathematical definitions of slippage
expressed directly in terms of the AMM's formula:
\emph{linear slippage} focuses on the buyer's price difference,
and \emph{angular slippage} characterizes how that buyer affects prices for later buyers.

Another cost is \emph{divergence loss} (sometimes called \emph{impermanent loss}),
where the value of the liquidity providers' investments end up worth less than if
the invested assets had been left untouched.
We give a precise mathematical definition of divergence loss,
expressed in directly terms of the AMM's formula,

We introduce a new figure of merit, called \emph{load},
that measures how costs are balanced among parties on different sides of a trade. 

We show how
these costs can be analyzed in either worst-case or in expectation.
We identify various \emph{conservation laws} that govern these costs:
they can be shifted, but never fully eliminated.
We characterize how these costs behave under sequential and parallel \emph{composition},
showing how to compute these costs for networks of AMMs,
not just individual AMMs.
Finally, we propose novel AMM designs capable of adapting to
changes in these costs.

The paper is organized as follows.
\secref{definitions} describes our model and terminology.
\secref{properties} introduces our cost measures and their conservation laws.
\secref{sequential} shows how these measure are affected
by sequential composition,
where the output of one AMM becomes the input of another AMM.
\secref{sequential} shows how these measure are affected
by parallel composition,
where traders split their trade between two AMMs that trade the same assets but according to different formulas.
\secref{oracle} surveys some simple adaptive strategies that
can mitigate the costs' conservation laws.
\secref{related} surveys related work.

Some of our numbered equations require long, mostly routine derivations which have been moved to the appendix to save space.
A few of the longer proofs have also been moved to the appendix.

\section{Definitions}
\seclabel{definitions}
We use bold face for vectors $(\bx)$ and italics for scalars ($x$).
Variables, scalar or vector,
are usually taken from the end of the alphabet ($x,y,z$),
and constants from the beginning ($a,b,c$).
We use ``:='' for definitions and ``='' for equality.
A function $f: \Reals \to \Reals$ is \emph{strictly convex}
if for all $t \in (0,1)$ and distinct $\bx,\bx' \in \Reals$,
$f(t \bx + (1-t) \bx') < t f(\bx) + (1-t) f(\bx')$.
Any tangent line for a strictly convex function lies below its curve.

Here is an informal example of a \emph{constant-product} AMM~\cite{uniswapv1}.
An AMM in state $(x,y)$ has custody $x$ units of asset $X$,
and $y$ units of asset $Y$,
subject to the invariant that the product $x y = c$,
for $x,y > 0$ and some constant $c > 0$.
The AMM's states thus lie on the hyperbolic curve $x y = c$.
If a trader transfers $\delta_X$ units of $X$ to the AMM,
the AMM will return $\delta_Y$ units of $Y$ to the trader,
where $\delta_Y$ is chosen to preserve the invariant $(x+\delta_X)(y-\delta_Y) = c$.

Formally,
the state of an AMM that trades assets $X$ and $Y$ is given by
a pair $(x,y) \in\PosReals^2$,
where $x$ is the number of $X$ units in the AMM's pool,
and $y$ the number of $Y$ units.
The state space is given by a curve $(x,f(x))$, where $f: \PosReals \to \PosReals$.
Except when noted,
the AMMs considered here satisfy the boundary
conditions $\lim_{x \to 0}f(x) = \infty$ and $\lim_{x \to \infty}f(x) = 0$, meaning that traders cannot exhaust either pool of assets.
The function $f$ is subject to further restrictions discussed later.

There are two kinds of participants in decentralized finance.
(1) \emph{Traders} transfer assets to AMMs, and receive assets back.
Traders can compose AMMs into networks to conduct more complicated
trades involving multiple kinds of assets.
(2) \emph{Liquidity providers} (or ``providers'') fund the AMMs by
lending assets, and receiving shares, fees, or other profits.
Traders and providers play a kind of alternating game:
traders modify AMM states by trading one asset for another,
and providers can respond by adding or removing assets,
reinvesting fees, or adjusting other AMM properties.

AMM typically charge fees for trades.
For example,
Uniswap v1 diverts 0.3\% of the assets returned by each trade
back into that asset's pool.
Although there is no formal difficulty including fees in our analysis,
we neglect them here because
they have little impact on costs:
fees slightly reduce both slippage costs for traders
and divergence loss for providers.

A \emph{valuation} $v \in (0,1)$ assigns relative values to
an AMM's assets: $v$ units of $X$ are deemed worth $(1-v)$ units of $Y$.
At valuation $v$,
a trader who moves an AMM from state $(x,f(x))$ to $(x',f(x'))$
makes a profit if $v(x-x')+(1-v)(f(x)-f(x'))$ is positive,
and otherwise incurs a loss.
The trader's profit is maximal precisely when $v x' + (1-v)f(x')$ is minimal.
We assume that at any time,
there is a single \emph{market valuation} accepted by most traders.
An \emph{arbitrage trade} is one in which a trader makes a profit
by moving an AMM from a state reflecting a prior market valuation to
a distinct state reflecting the current market valuation.

A \emph{stable point} for an AMM $A$ and valuation $v$
is a point $(x,f(x))$ that minimizes $v x + (1-v)f(x)$.
If $v$ is the market valuation,
then any trader can make an arbitrage profit
by moving the AMM from any state to a stable state,
and no trader can make a profit by moving the AMM out of a stable state.
Valuations, stable points, and exchange rates are related.
If $(x,f(x))$ is the stable point for valuation $v$, then
\begin{equation}
    \frac{d f(x)}{d x} = -\frac{v}{1-v}.
\end{equation}

Following Engel and Herlihy~\cite{EngelH2021},
we require the AMM function $f$ to satisfy certain reasonable properties,
expressed here as axioms.
A detailed discussion and justification for each axiom appears elsewhere~\cite{EngelH2021}.

Every AMM state should define a unique rate of exchange
between its assets,
and trades should change that rate gradually rather than abruptly.
\begin{axiom}[Continuity]
  \axiomlabel{continuity}
The function $f$ is strictly decreasing and (at least) twice-differentiable.
\end{axiom}

An AMM must be able to adapt to any market conditions.
The \emph{exchange rate} of asset $Y$ in units of $X$ at state $(x,f(x))$
is $-f'(x)$, the negative of the curve's slope at that point.
\begin{axiom}[Expressivity]
  \axiomlabel{expressivity}
  The exchange rate $-f'(x)$ can assume every value
  in the open interval $(0,\infty)$.
\end{axiom}

Slippage should work to the disadvantage of the trader.
To prevent runaway trading,
buying more of asset $X$ should make $X$ more expensive, not less.
\begin{axiom}[Convexity]
\axiomlabel{convexity}
For every AMM $A:=(x,f(x))$, $f$ is strictly convex.
\end{axiom}

It can be shown~\cite{EngelH2021} that for any AMM satisfying these axioms,
every valuation has a unique stable state.
\begin{theorem}[Stability]
\thmlabel{stability}
The function
\begin{equation}
    \phi(v) = f^{\prime -1}\left(-\frac{v}{1-v}\right) 
\end{equation}
is a homeomorphism $\phi: (0,1) \to \PosReals$ that carries
each valuation $v \in (0,1)$
to the unique stable state $x$ that minimizes $v x + (1-v)f(x)$.
\end{theorem}
For example, the stable state map for the constant-product AMM $(x,1/x)$
is $\phi(v) = \sqrt{\frac{(1-v)}{v}}$.
Sometimes it is convenient to express $\phi$ in vector form as
$\Phi: (0,1) \to \PosReals^2$,
where $\Phi(v) := (\phi(v),f(\phi(v)))$.
We will sometimes use $\psi: \PosReals \to (0,1)$,
the inverse function of $\phi$:
\begin{equation}
\psi(x) = -\frac{f'(x)}{1-f'(x)}.
\end{equation}
The vector form is $\Psi(x) := (\psi(x),1-\psi(x))$.

Most of the properties of interest in this paper can be expressed either in the asset domain,
as functions of $x$ and $f(x)$,
or in the valuation domain, as functions of $v$ and $1-v$.

These definitions extend naturally to AMMs that trade more than two assets.
Many (but not all) if the results presented below also generalize,
but for brevity we focus on AMMs that trade between two assets.

\section{Properties of Interest}
\seclabel{properties}

\subsection{AMM Capitalization}
\seclabel{capital}
\begin{figure}
\centering
  \includegraphics[width=0.5 \hsize]{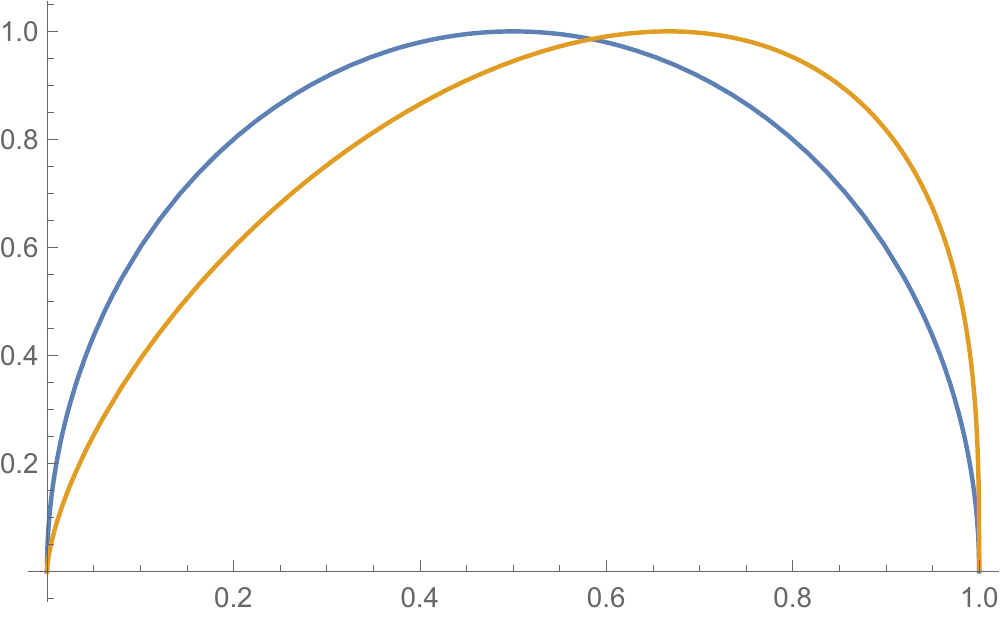}
  \caption{Balanced Capitalization of AMMs $(x,1/x)$ (symmetric curve, blue) and $(x,1/x^2)$ (asymmetric curve, yellow) at stable points}
  \figlabel{cap}
\end{figure}
Let $(v,1-v)$ be a valuation with stable point $(x,f(x))$. 
What is a useful way to define the \emph{capitalization} (total value) of an AMM's holdings?
It may be appealing to pick one asset to act as \emph{num\'eraire},
computing the AMM's capitalization at point $(x,f(x))$ in terms of that asset alone:
\begin{equation*}
  \numcap_X(v;A) := \phi(v) + \left(\frac{1-v}{v}\right) f(\phi(v)).
\end{equation*}
For example,
the num\'eraire capitalization of the constant-product AMM 
at the stable point for valuaton $v$ is
$2\sqrt{\frac{1}{v}-1}$.

Unfortunately,
this notion of capitalization can lead to counter-intuitive results
if the num\'eraire asset becomes volatile.
As the value of the num\'eraire  $X$ tends toward zero,
``bad money drives out good'',
and arbitrage traders will replace more valuable $Y$ with less valuable $X$.
As the AMM fills up with increasingly worthless units of $X$,
its num\'eraire capitalization grows without bound,
so an AMM whose holdings have become worthless has infinite num\'eraire capitalization.

A more robust approach is to choose a formula that balances the asset classes in proportion to their relative valuations.
Define the (balanced) \emph{capitalization} at $(x,f(x))$ to be the sum of the two asset pools weighted by their relative values.
Let $\bv = (v,1-v)$ and $\bx = (x,f(x))$.
\begin{equation}
  \capital(x,v; A) := v x + (1-v) f(x)  = \bv \cdot \bx.
\end{equation}
If $v$ is the current market valuation,
then $A$ will usually be in the corresponding stable state $\Phi(v) = (\phi(v),f(\phi(v))$, yielding
\begin{equation}
  \capital(v; A) := \capital(\phi(v),v; A) = \bv \cdot \Phi(v).
\end{equation}

For example, AMMs $A:=(x,1/x)$ and $B:=(x,1/x^2)$
have capitalization at their stable points:
\begin{align*}
  \capital(v; (x,1/x)) &= 2\sqrt{v(1-v)}, \\
  \capital(v; (x,1/x^2)) &= \frac{3 v\sqrt[3]{\frac{1}{v}-1}}{2^{2/3}}.
\end{align*}
See \figref{cap}.
Both have minimum capitalization 0 at $v=0$ and $v=1$,
where one asset is worthless,
and maximal capitalizations at respective valuations $v=1/2$ and $v= 2/3$.

\begin{theorem}
\thmlabel{max-cap}
An AMM's capitalization is maximal at the fixed-point $\phi(v) = f(\phi(v))$, where the amounts of $X$ and $Y$ are equal.
\end{theorem}

However, \Figref{cap} shows that the valuation that maximizes an AMM's capitalization is not necessarily $\frac{1}{2}$.
An AMM is \emph{symmetric} if $f = f^{-1}$.
(Both Uniswap and Curve use symmetric curves.)
\begin{lemma}
\lemmalabel{sym-curve}
The stable state map for any symmetric AMM satisfies $f(\phi(v)) = \phi(1-v)$
\end{lemma}
\begin{theorem}
\thmlabel{sym-max-cap}
    Any symmetric AMM has maximum capitalization at $\bv = (\frac{1}{2},\frac{1}{2})$.
\end{theorem}

\subsection{Divergence Loss}
Consider the following simple game.
A liquidity provider funds an AMM $A$,
leaving it in the stable state $\bx$ for the current market valuation $\bv$.
Suppose the market valuation changes from $\bv$ to $\bv'$,
and a trader submits an arbitrage trade that would take $A$ to the stable state $\bx'$ for the new valuation.
The provider has a choice:
(1) immediately withdraw its liquidity from $A$ instead of accepting the trade,
or (2) accept the trade.
It is not hard to see that
the provider should always choose to withdraw.
The shift to the new valuation changes the AMM's capitalization.
If $\bx$ is the stable state for $\bv$,
then it cannot be stable for $\bv'$,
so the difference between the capitalizations $\bv' \cdot \bx - \bv' \cdot \bx'$ must be positive,
so the arbitrage trader would profit at the provider's expense.
(In practice, a provider would take into account the value of current and future fees in making this decision.)

Define the \emph{divergence loss} for an AMM $A = (x,f(x))$ to be
\begin{align*}
  \divloss(v,v';A)
  &:= \bv' \cdot \Phi(v) - \bv' \cdot \Phi(v') \\
  &=  v' \phi(v) + (1-v')f(\phi(v)) - (v' \phi(v') + (1-v')f(\phi(v')))
\end{align*}
where $\Phi(v,1-v) = (\phi(v),f(\phi(v)))$.
Sometimes it is useful to express divergence loss in the
trade domain, as a function of liquidity pool size instead of valuation:
\begin{align*}
  \divloss^*(x,x';A)
  &:= \divloss(\psi(x),\psi(x');A)\\
  &= \Psi(x') \cdot \bx - \Psi(x') \cdot \bx' \\
  &= \psi(x') x + (1-\psi(x'))f(x) - (\psi(x') x' + (1-\psi(x'))f(x').
\end{align*}
Informally, divergence loss measures the difference in value between funds invested in an AMM and funds left in a wallet.
Recall that by definition $\bv'$ minimizes $\Phi(v') \cdot \bv'$,
so divergence loss is always positive when $\bv \neq \bv'$.

Is it possible to bound divergence loss by bounding trade size?
More precisely,
can an AMM guarantee that any trade that adds $\delta$ or fewer units of $X$
incurs a divergence loss less than some $\epsilon$, for positive constants $\delta,\epsilon$?

Unfortunately, no.
There is a strong sense in which divergence loss can be shifted, but never eliminated.
For example, for the constant-product AMM $A:=(x,1/x)$,
the divergence loss for a trade of size $\delta$ is
\begin{equation}
\divloss^*(x,x+\delta; A) = \frac{\delta ^2}{2 \delta  x^2+x^3+\delta ^2 x+x}
\end{equation}
Holding $\delta$ constant and letting $x$ approach 0,
the divergence loss for even constant-sized trades grows without bound.
This property holds for all AMMs.
\begin{theorem}
\thmlabel{divloss-conserve}
No AMM can bound divergence loss even for bounded-size trades.
\end{theorem}
  
\begin{proof}
For AMM $A:=(x,f(x))$,
\begin{multline*}
\divloss^*(x,x+\delta; A) = \psi(x+\delta) x + (1-\psi(x+\delta)) f(x) \\
- (\psi(x+\delta) (x+\delta) + (1-\psi(x+\delta)) f(x+\delta)).
\end{multline*}
Note that $\lim_{x \to 0} (1-\psi(x+\delta)) > 0$,
and $\lim_{x \to 0} f(x) = \infty$.
All other terms have finite limits,
so $\lim_{x \to 0}\divloss^*(x,x+\delta; A) = \infty.$
\end{proof}

What is a provider's worst-case exposure to divergence loss?
Consider an AMM $A:=(x,f(x))$ in state $(a,f(a))$.
As the $X$ asset becomes increasingly worthless,
the valuation $v'$ approaches $(0,1)$
as its stable state approaches $(\infty,0)$.
\begin{equation*}
  \lim_{v' \to 0} \divloss(v,v'; A)
    = (v',1-v') \cdot (a,f(a))
    - (v',1-v') \cdot (\phi(v'),f(\phi(v')))) = f(a).
\end{equation*}
Symmetrically, if the $Y$ asset becomes worthless,
\begin{equation*}
  \lim_{v' \to 1} \divloss(v,v'; A)
    = (v',1-v') \cdot (a,f(a))
    - (v',1-v') \cdot (\phi(v'),f(\phi(v')))) =a.
\end{equation*}
The provider's worst-case exposure to divergence loss in state $(a,f(a))$ is thus $\max(a,f(a))$.
The \emph{minimum} worst-case exposure occurs when $a = f(a)$.
Recall from \secref{capital} that this fixed-point is exactly the state that maximizes
the AMM's capitalization.
(See Appendix \secref{min-divcap} for a more formal treatment of this claim.)

Divergence loss is sometimes called \emph{impermanent loss},
because the loss vanishes if the assets return to their original valuation.
The inevitability of impermanent loss does not mean that an AMM's capitalization cannot increase,
only that there is always an opportunity cost to the provider for not cashing in earlier.

\subsection{Linear Slippage}
Linear slippage measures how increasing the size of a trade diminishes that trader's rate of return.
Let $A := (x,f(x))$ be an AMM in stable state $(a,f(a))$ for valuation $v$.
Suppose a trader sends $\delta >0$ units of $X$ to $A$,
taking $A$ from $(a,f(a))$ to $(a+\delta,f(a+\delta))$,
the stable state for valuation $v' < v$.
If the rate of exchange were linear,
the trader would receive $-\delta f'(a)$ units of $Y$ in return for $\delta$ units of $X$.
In fact, the trader receives only $f(a) - f(a+\delta)$ units,
for a difference of $-\delta f'(a) - f(a) + f(a+\delta)$.

The \emph{linear slippage}(with respect to $X$) is the value of this difference:
\begin{equation}
\eqnlabel{linslipx}
  \linslip_X(v,v'; A) = \left(\frac{1-v'}{1-v}\right)\left( \bv \cdot \Phi(v') - \bv \cdot \Phi(v) \right)
\end{equation}
In the trade domain, $\linslip^*_X(x,x'; A) = \linslip_X(\psi(x),\psi(x'); A)$.
Linear slippage with respect to $Y$ is defined symmetrically.

For example,
the linear slippage for the constant-product AMM $A:=(x,1/x)$
for a trade of size $\delta$ is
\begin{equation}
\linslip^*_X(x,x+\delta; A) = -\frac{\delta ^2 (\delta +x)}{x^2 \left(\delta ^2+x^2+2 \delta  x+1\right)}
\end{equation}
Just as for divergence loss,
holding $\delta$ constant and letting $x$ approach 0,
the linear slippage across constant-sized trades grows without bound.
\begin{theorem}
No AMM can bound linear slippage for bounded-size trades.
\end{theorem}

\begin{proof}
As in the proof of \thmref{divloss-conserve},
the claim follows because 
$\lim_{x->0} \linslip^*_X(x,x+\delta; A) = \infty$.
\end{proof}

\subsection{Angular Slippage}
\begin{figure}[t]
\centering
  \includegraphics[width=0.5 \hsize]{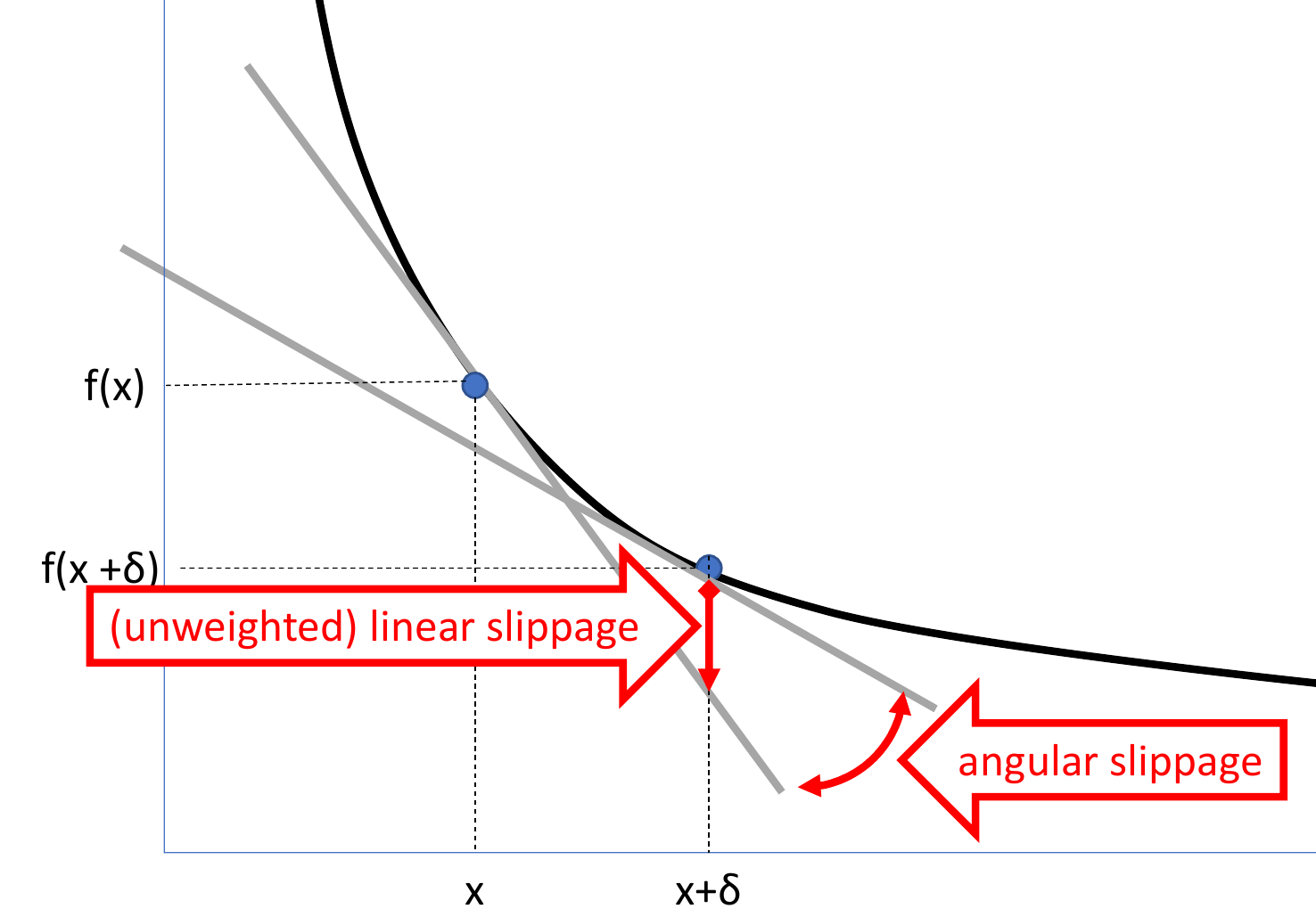}
  \caption{Angular vs linear slippage)}
  \figlabel{slip}
\end{figure}
Angular slippage measures how the size of a trade affects the exchange rate between the two assets.
This measure focuses on how a trade affects the traders who come after.
Recall that the (instantaneous) exchange rate in state $(x,f(x))$ is the slope of the tangent $f'(x)$.
Let $\theta(x)$ denote the angle of that tangent with the $X$-axis.
(We could equally well use the tangent's angle with the $Y$-axis.)
A convenient way to measure the change in price is to measure the change in angle.
Consider a trade that carries $A$ from 
valuation $v$ with stable point $(x,f(x))$,
to valuation $v'$ with stable point $(x',f(x'))$.
Define the \emph{angular slippage} of that trade
to be the difference in tangent angles at $x$ and $x'$
(expressed in the valuation domain):
\begin{equation*}
\angslip(v,v';A) = \theta(\phi(v'))-\theta(\phi(v)).  
\end{equation*}
In the trade domain,
$\angslip^*(x,x'; A) := \angslip(\phi(x),\phi(x');A)$.
Angular slippage is \emph{additive}:
for distinct valuations $v < v' < v''$,
$
\angslip(v,v'';A) = \angslip(v,v';A) + \angslip(v',v'';A).
$
Note that linear slippage is not additive.

Angular slippage and linear slippage are different ways
of measuring the same underlying phenomenon:
their relation is illustrated in \figref{slip}.

Here is how to compute angular slippage.
By definition,$ \tan \theta(x) = -\frac{1}{f'(x)}$.
Let $x' > x$, $\bv = (v,1-v)$, and $\bv' = (v',1-v')$.
\begin{equation}
\eqnlabel{angslip}
    \angslip(v,v';A) = \theta(x') - \theta(x) = \arctan\left(\frac{v-v'}{\bv \cdot \bv'}\right)
\end{equation}
The next lemma says says that the overall angular slippage,
$\angslip(0,\infty)$, is a constant independent of the AMM.
\begin{theorem}
\thmlabel{angular}
For every AMM $A$, $\angslip(0,\infty; A) = \pi/2$.
\end{theorem}

\begin{proof}
    Consider an AMM $A:=(x,f(x))$.
    As 
     $\lim_{x \to 0} f'(x) = - \frac{v}{1-v} = -\infty$,
    implying $v = 1$.
    As
    $\lim_{x \to \infty} f'(x) = - \frac{v}{1-v} = 0$
    implying $v = 0$.
    Because $\tan \angslip(0,\infty; A) = \frac{1 - 0}{(0,1) \cdot (0,1)} = \infty$,
    $\angslip(0,\infty; A) = \frac{\pi}{2}$.
\end{proof}

The additive property means that no AMM can eliminate angular slippage over every finite interval.
Lowering angular slippage in one interval requires increasing it elsewhere.
\begin{corollary}
For any AMM $A$,
and any level of slippage $s$, $0 < s < \pi/2$,
there is an interval $(x_0,x_1) \subset \PosReals$ such that $\angslip^*(x_0,x_1; A) > s$.
\end{corollary}
For example,
the Curve~\cite{curve} AMM advertises itself as having lower slippage than its competitors.
\thmref{angular} helps us understand this claim:
compared to a constant-product AMM,
Curve does have lower slippage than a constant-product AMM
for stable coins when they trade at near-parity,
but it must have higher slippage when the exchange rate wanders out of that interval.

\subsection{Load}
Divergence loss is a cost to providers,
and linear slippage is a cost to traders.
Controlling one without controlling the other is pointless
because AMMs function only if both providers and traders consider their costs acceptable.
We propose the following measure to balance provider-facing and
trader-facing costs.
The \emph{load} (with respect to $X$) across an interval is
the product of that interval's divergence loss and linear slippage:
\begin{equation}
\load_X(v,v';A) := \divloss(v,v'; A) \linslip_X(v,v'; A)
\end{equation}
Load can also be expressed in the trade domain:
$\load^*_X(x,x';A) := \divloss^*(x,x'; A) \linslip^*_X(x,x'; A)$.

\subsection{Expected Load}
\begin{figure}[t]
\centering
  \includegraphics[width=0.6 \hsize]{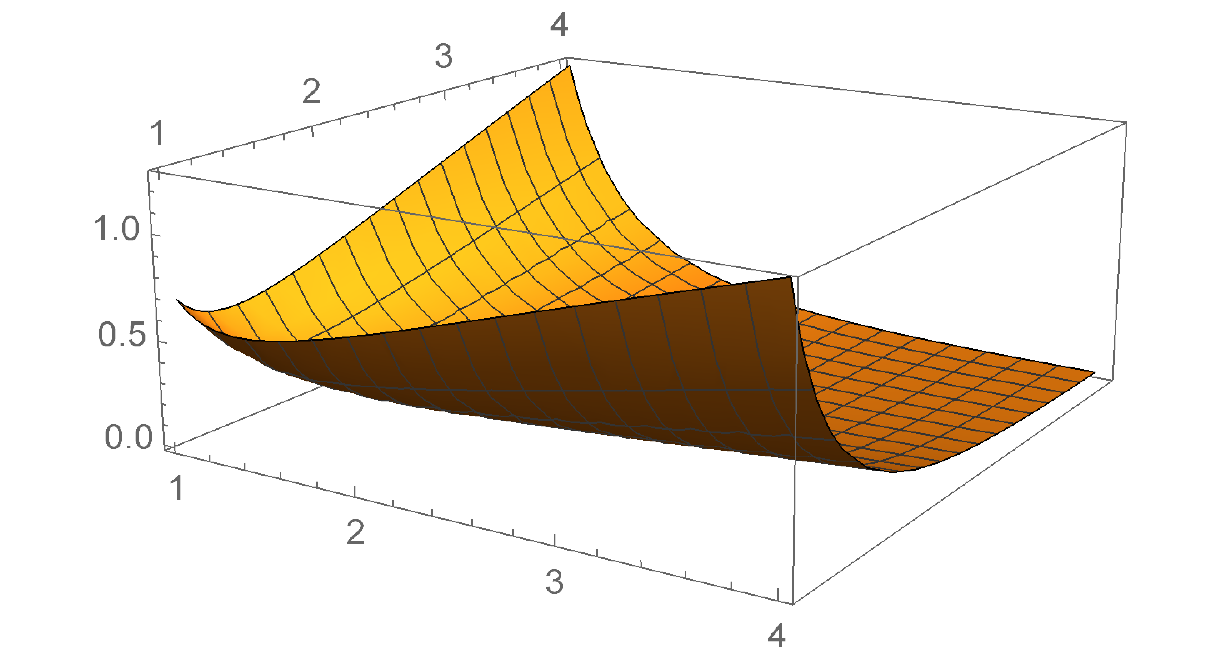}
  \caption{Expected load for $A:=(x,1/x)$ at 1/2 as a function of $\beta$ distribution parameters. (Mathematica source is shown in Appendix \secref{mathematica}.)}
  \figlabel{expload}
\end{figure}
We have seen that cost measures such as divergence loss,
linear slippage, angular slippage, and load
cannot be bounded in the worst case.
Nevertheless, these costs can be shifted.
Not all AMM states are equally likely.
For example, one would expect stablecoins to trade at near parity~\cite{curve}.

Suppose we are given a probability density for future valuations.
This distribution might be given \emph{a priori},
or it may be learned from historical data.
Can we compare the behavior of alternative AMMs given such a distribution?

Let $p(v)$ be the distribution over possible future valuations.
The expected load when trading $X$ for $Y$
starting in the stable state for valuation $\bv$ is
\begin{equation*}
  \int_0^{v}P(v'| v' < v)\load_X(v,v';A)dv' 
\end{equation*}
Weighting this expectation with the probability $P(v>v')$
that the trade will go in that direction yields
\begin{equation*}
  P(v>v') \int_0^{v}P(v'| v' < v)\load_X(v,v';A) dv' =
  \int_0^{v}p(v')\load_X(v,v';A) dv'.
\end{equation*}
Define the \emph{expected load} of AMM $A$ at valuation $v$ to be:
\begin{equation*}
  E_p[\load(v;A)] :=
  \int_0^{v}p(v')\load_X(v,v';A)dv' +
  \int_v^{1}p(v')\load_Y(v,v';A)dv'.
\end{equation*}
Of course, one can compute the expected value of any the measures proposed here, not just load.

\Figref{expload} shows the expected load for $A:=(1,1/x)$ ,
starting at valuation $1/2$,
where the expectation is taken over the $\betaDist(\alpha_1,\alpha_2)$ distributions~\cite{betadist},
where parameters $\alpha_1,\alpha_2$ range independently from 1 to 4.
Inspecting the figure shows that
symmetric distributions $\beta(\alpha,\alpha)$,
which are increasingly concentrated around $1/2$ as $\alpha$ grows,
yield decreasing loads as the next valuation becomes increasingly likely to be close to the current one.
By contrast, asymmetric distributions, which favor unbalanced valuations,
yield higher loads because the next valuation is likely to be farther from the current one.


\section{Sequential Composition}
\seclabel{sequential}
The \emph{sequential composition} of two AMMs
is constructed by using the output of one AMM as the input to the other.
(See Engel and Herlihy~\cite{EngelH2021} for a proof
that the sequential composition of two AMMs is an AMM.)
For example, if $A$ trades between florins and guilders,
and $B$ trades between guilders and francs,
then their sequential composition $A \oplus B$ trades between florins and francs.
A trader might deposit florins in $A$, receiving guilders,
then deposit those guilders in $B$, receiving francs.
In this section,
we investigate how linear slippage and divergence loss
interact with sequential composition.

Consider two AMMs $A = (x,f(x)),B = (y,g(y))$,
where $A$ trades between $X$ and $Y$,
and $B$ between $Y$ and $Z$.
If $A$ is in state $(a,f(a))$ and $B$ in state $(b,g(b))$
then their sequential composition is $A \oplus B := (x,h(x))$,
where $h(x)= g(b+f(a)-f(x))$~\cite{EngelH2021}.
(The sequential composition of more than two AMMs can be constructed by repeated two-way compositions.)

Let $\bv = (v_1,v_2,v_3)$ be the market valuation linking $X,Y,Z$,
inducing pairwise valuations
\begin{equation*}
  v_{12} = \frac{v_1}{v_1+v_2},\;
  v_{23} = \frac{v_2}{v_2+v_3},\;
  v_{13} = \frac{v_1}{v_1+v_3},\\
\end{equation*}
along with their vector forms
$
  \bv_{12} = (v_{12}, 1-v_{12}),\;
  \bv_{23} = (v_{23}, 1-v_{23}),\;
  \bv_{13} = (v_{13}, 1-v_{13})$.
Let $\bv' \neq \bv$ be a three-way valuation inducing 
analogous pair-wise valuations.
Let $\phi_A,\phi_B,\phi_{AB}: (0,1) \to \PosReals$ be the stable point maps for $A,B,A \oplus B$ respectively,
and $\Phi_A,\Phi_B,\Phi_{AB}: (0,1) \to \PosReals^2$ their vector forms.

Our composition rules apply when $A$ and $B$ start in
their respective stable states\footnote{
If $A$ and $B$ do not start in stable states for the current market valuation, 
then an arbitrage trader will eventually put them there.}
for a market valuation $\bv$:
$(a,f(a))$ is the stable state for $v_{12}$,
$(a+\delta,f(a+\delta))$ for $v_{12}'$.
$(b,g(b))$ for $v_{23}$,
and $(b+f(a)-f(a+\delta),g(b+f(a)-f(a+\delta))$ for $v_{23}'$.
We analyze the changes in divergence loss and linear slippage when
the market valuation changes from $\bv$ to $\bv'$.

\subsection{Divergence Loss}
Initially, the combined capitalization of $A$ and $B$ is $v_1 x + v_2(f(x) + y) + v_3 g(y)$.
A trader sends $\delta$ units to $A$,
reducing the combined capitalization by 
\begin{equation}
    \eqnlabel{divlossxy}
    -\delta v_1' + v_2'(f(x) - f(x + \delta)) = v_3'\divloss(v_{12},v_{12}'; A) 
\end{equation}
Next the trader sends the assets returned from the first trade to $B$,
reducing the combined capitalization by:
\begin{equation}
    \eqnlabel{divlossyz}
    v_2'(f(x + \delta) - f(x)) + v_3'(g(y) - g(y + f(x) - f(x + \delta))) = v_1' \divloss(v_{23},v_{23}'; B)
\end{equation}
Finally,
treating both trades as a single transaction reduces the combined capitalization by:
\begin{equation}\eqnlabel{divlossxz}
    -\delta v_1' + v_3'(h(x) - h(x + \delta))= v_2'\divloss(v_{13},v_{13}'; A \oplus B)
\end{equation}
Combining \eqnrange{divlossxy}{divlossxz} yields
\begin{equation}
\eqnlabel{divloss-seq}
\divloss(v_{13}, v_{13}'; A \oplus B)
  = \left( \frac{v_3'}{v_2'} \right) \divloss(v_{12}, v_{12}'; A) +
  \left( \frac{v_1'}{v_2'} \right) \divloss(v_{23}, v_{23}'; B).
\end{equation}
The effect of sequential composition on divergence loss is linear but not additive:
the divergence loss of the composition is a weighted sum
of the divergence losses of the components.

\subsection{Linear Slippage}
With respect to $\bv$,
a trader who sends $\delta$ units of $X$ to $A$ incurs the following slippage
\begin{equation}
 \eqnlabel{linslipxy}  
 v_2'(-\delta f'(x) + f(x + \delta) - f(x))
 = (v_1'+v_2')\linslip_X(v_{12}',v_{12}; A).
\end{equation}
Next the trader sends the assets returned from the first trade to $B$,
incurring the following slippage:
\begin{equation}
    \eqnlabel{linslipyz}
    v_3' \left(
(f(x)-(x+\delta))g'(y)
+ g(y+f(x)-(x+\delta))
- g(y)
\right)=
(v_2'+v_3') \linslip_X(v_{23},v_{23}'; B).
\end{equation}

Finally,
treating both trades as a single transaction yields slippage:
\begin{equation}
    \eqnlabel{linslipxz}
    v_3'(-\delta h'(x) + h(x + \delta) - h(x))
     (v_1'+v_3')\linslip_X(v_{13}',v_{13}, A \oplus B)
\end{equation}

Combining \eqnrange{linslipxy}{linslipxz} yields
\begin{equation}
\linslip_X(v_{13},v_{13}'; A \oplus B) =
\left(\frac{1-v_3'}{1-v_2'}\right)\linslip_X(v_{12},v_{12}'; A) +
\left(\frac{1-v_1'}{1-v_2'}\right)\linslip_X(v_{23},v_{23}'; B)
    \eqnlabel{linslip-seq}
\end{equation}

\subsection{Angular Slippage}
A trader sends $\delta$ to $A$,
where $f(a+\delta)$ is the stable point for $v_{12}'$,
$g(b+f(a)-f(a+\delta))$ the stable point for $v_{23}'$.
By construction,
\begin{gather*}
  h'(a) = -g'(b)f'(a) = -\frac{v_{23}}{1-v_{23}}\frac{v_{12}}{1-v_{12}}
  = -\frac{v_2}{v_3}\frac{v_1}{v_2} = -\frac{v_1}{v_3}\\
  h'(a+\delta) = -g'(b+f(a)-f(a+\delta))f'(a+\delta)
  =
  -\frac{v_{23}'}{1-v_{23}'}\frac{v_{12}'}{1-v_{12}'}
  =
  -\frac{v_2'}{v_3'}\frac{v_1'}{v_2'}
  =
  -\frac{v_1'}{v_3'}
\end{gather*}

Define $\theta_A(x),\theta_B(y),\theta_B(y)$ to be the 
respective angles of $f'(x)$, $g'(y)$, and $h'(x)$ with their $X$-axes.
We can express the tangents of the composite AMM's angles
in terms of the tangents of the component AMMs' angles.
\begin{gather*}
  \tan \theta_{AB}(x)
  = -\frac{1}{h'(x)}
  = -\frac{1}{-g'(b+f(a)-f(x))f'(x)}\\
  \tan \theta_{AB}(a)
  =\frac{v_3}{v_2}\frac{v_2}{v_1}
  =\frac{v_3}{v_1},\qquad
  \tan \theta_{AB}(a+\delta)
  =\frac{v_3'}{v_2'}\frac{v_2'}{v_1'}
  =\frac{v_3'}{v_1'}.
\end{gather*}
The component AMMs $A$ and $B$ determine the valuations $v_1,v_2,v_1',v_2'$,
which induce the remaining valuations $v_3,v_3',v_{12},v_{23},v_{13},v_{12}',v_{23}',v_{13}'$.
\begin{equation}
    \eqnlabel{angslip-seq}
    \angslip(v_{13},v_{13'}, A \oplus B) = \arctan \left(\frac{v_{13}-v_{13}'}{\bv_{13} \cdot \bv_{13}'}\right)
\end{equation}
It follows that the angular slippage of the sequential composition of two AMMs can be computed from the component AMMs' valuations.

\subsection{Load}
Combining \eqnref{divloss-seq} and \eqnref{linslip-seq} yields
\begin{align}
    \eqnlabel{load-seq}
    \load_X(v_{13},v_{13}'; A \oplus B) &= 
  \left( \frac{v_3'(1-v_3')}{v_2'(1-v_2')} \right) \load_X(v_{12}, v_{12}'; A)
  +
  \left( \frac{v_1'(1-v_1')}{v_2'(1-v_2')} \right) \load_X(v_{23}, v_{23}'; B)\\
  &\qquad+
  \left( \frac{v_3'(1-v_3')}{v_2'(1-v_2')} \right)
  \divloss(v_{12}, v_{12}'; A)
  \linslip_X(v_{23},v_{23}'; B)\notag\\
  &\qquad+
  \left( \frac{v_1'(1-v_1')}{v_2'(1-v_2')} \right)
  \divloss(v_{23}, v_{23}'; B)
  \linslip_X(v_{12},v_{12}'; A)\notag
\end{align}
It follows that the load of a sequential composition is a weighted sum of the loads of the components,
plus additional (strictly positive) cross-terms.

\section{Parallel Composition}
\seclabel{parallel}
Parallel composition~\cite{EngelH2021} arises when a trader is presented with two AMMs
$A:=(x,f(x))$ and $B:=(x,g(x))$,
both trading assets $X$ and $Y$,
and seeks to treat them as a single combined AMM $A||B$.
Let $\phi_A,\phi_B: (0,1) \to \PosReals$ be the stable point maps for $A,B$ respectively,
with $\Phi_A,\phi_B: (0,1) \to  \PosReals^2$ their vector forms,
and $\psi_A,\psi_B: \PosReals \to (0,1)$ their inverses.
(The parallel composition of more than two AMMs can be constructed by
repeated two-way compositions.)

As shown elsewhere~\cite{EngelH2021},
a trader who sends $\delta$ units of $X$ to the combined AMM
maximizes return by splitting those units between $A$ and $B$,
sending $t \delta$ to $A$ and $(1-t) \delta$ to $B$,
for $0 \leq t \leq 1$,
where
$f'(x + t\delta) = g'(y + (1-t)\delta)$.

We assume traders are rational, and always split trades in this way.
Because the derivatives are equal,
$x + t\delta$ and $y + (1-t)\delta$ are stable points
of $A$ and $B$ respectively for the same valuation $v'$,
so $x + t\delta = \phi_A(v')$, and $y + (1-t)\delta = \phi_B(v')$.
If $A$ is in state $(a,f(a))$ and $B$ in $(b,g(b))$,
then $A||B := (x,h(x))$, where
$h(x) := (f(a) - f(a+t x)+ g(b) - g(b+(1-t)x)$.
Our composition rules apply when both $A,B$ are in their stable states
for valuation $\bv = (v,1-v)$.
We analyze the change in divergence loss and linear slippage
when the common valuation changes from $\bv$ to $\bv'=(v',1-v')$.
The new valuation $\bv'$ may be the new market valuation,
or it may be the best the trader can reach with a fixed budget of $\delta$.

\subsection{Divergence Loss}
If a trader sends $\delta$ units of $X$ to $A || B$,
the combined capitalization suffers a loss of 
\begin{equation}
    \eqnlabel{divloss-par}
    \divloss(v,v';A||B) = \divloss(v,v';A) + \divloss(v,v';B).
\end{equation}
It follows that divergence loss under parallel composition is additive.
\subsection{Linear Slippage}
A straightforward calculation shows:
\begin{equation}
\eqnlabel{linslip-par}
    \linslip_X(v,v';A||B) = \linslip_X(v,v'; A) + \linslip_X(v,v'; B)
\end{equation}
Linear slippage is thus additive under parallel composition.

Linear slippage is also linear under scalar multiplication.
Any AMM $A:=(x,f(x))$ can be \emph{scaled} by a constant $\alpha  > 0$
yielding a distinct AMM $\alpha A := (\alpha x,\alpha f(x))$.
Let $(x,f(x))$ be the stable point for valuation $\bv$,
and $(x',f'(x'))$ the stable point for $\bv'$.
\begin{equation}
    \eqnlabel{linslip-scalar}
    \linslip_X(v,v';\alpha A) = \alpha \linslip_X(v,v';A).
\end{equation}

\subsection{Angular Slippage}
Because both $A$ and $B$ go from stable states for $v$
to stable states for $v'$,
\begin{gather*}
f'(a) = g'(a) = h'(a) = \frac{-v}{1-v},\quad
f'(a+t \delta) =g'(a+(1-t)\delta) = h'(a+\delta) = \frac{-v'}{1-v'}
\end{gather*}
It follows that
\begin{equation}
\eqnlabel{angslip-par}
  \angslip(v,v'; A||B) =
  \angslip(v,v'; A) =
  \angslip(v,v'; B).
\end{equation}

\subsection{Load}
Combining \eqnref{divloss-par} and \eqnref{linslip-par} yields
\begin{multline}
  \load_X(v,v'; A || B)
  =
  \load_X(v, v'; A)
  +
  \load_X(v, v'; B)\\
   \divloss(v, v'; A)
  \linslip_X(v,v'; B)
  \divloss(v, v'; B)
  \linslip_X(v,v'; A)
  \eqnlabel{load-par}
\end{multline}
It follows that the parallel composition's load
is the sum of the components' loads,
plus additional (strictly positive) cross-terms.

\section{Adaptive Strategies}
\seclabel{oracle}
So far we have proposed several ways to quantify the costs
associated with AMMs.
Now we turn our attention to strategies for adapting to
cost changes.
A complete analysis of adaptive AMM strategies is material for
another paper, so here we summarize two broad strategies
motivated by our proposed cost measures.
We focus on adjustments that might be executed automatically,
without demanding additional liquidity from providers.

\subsection{Change of Valuation}
Suppose an AMM learns, perhaps from a trusted Oracle service,
that its assets' market valuation has moved away from the AMM's current stable state,
leaving the providers exposed to substantial divergence loss.
Specifically,
suppose $A := (x,f(x))$ has valuation $v_1$ with stable state $(a_1,b_1)$,
when it learns that the market valuation has changed to $v_2$
with stable point $(a_2,b_2)$.

An arbitrage trader would move $A$ from $(a_1,b_1)$ to $(a_2,b_2)$,
pocketing a profit.
Informally,
$A$ can eliminate that divergence loss by ``pretending'' to conduct that arbitrage trade itself,
leaving the state the same, but moving the curve.
We call this strategy \emph{pseudo-arbitrage}.

$A$ changes its function using linear changes of variable in $x$ and $y$.
Suppose $a_1 > a_2$ and $b_2 > b_1$.
First, replace $x$ with $x-(a_1-a_2)$,
shifting the curve along the $X$-axis.
Next, replace $y$ with $y-(b_2-b_1)$,
shifting the curve along the $y$-axis.
The transformed AMM is now $A':= (x,f(x-(a_1-a_2))-(b_2-b_1))$.
The current state $(a_1,b_1)$ still lies on the shifted curve,
but now with slope $\frac{v_2}{v_2-1}$, matching the new valuation.
The advantage of this change is that $A$'s providers are no longer
exposed to divergence loss from the new market valuation.

The disadvantage is that pseudo-arbitrage produces AMMs that do not satisfy the usual boundary conditions 
$f(0) = \infty$ and $f(\infty) = 0$,
although they continue to satisfy the AMM axioms.
In practical terms, $A$ now has more units of $X$ than it needs,
but not enough units of $Y$ to cover all possible trades.
The AMM must refuse trades that would lower its $Y$ holdings below zero,
and there are $(a_1-a_2)$ units of $X$ inaccessible to the AMM.
The liquidity providers might withdraw this excess,
they might ``top up'' with more units of $Y$ to rebalance the pools,
or they might leave the extra balance to cover future pseudo-arbitrage changes.
(Note that $A$'s ability to conduct trades only while the valuation stays within a certain range is similar to Uniswap v3's ``concentrated liquidity'' option.)

\subsection{Change of Distribution}
Suppose an AMM's formula was initially chosen to match a predicted distribution on future valuations.
If that prediction changes,
then it may be possible to adjust the AMM's formula to match the new prediction.
Such an adjustment might be built into the AMM's smart contract,
or it could be imposed from outside by the liquidity providers.
The AMM's current function could be replaced with an alternative
that improves some expected cost measure,
say, reducing expected load or increasing expected capitalization.
But replacing AMM $A:=(x,f(x))$, in the stable state for the market valuation,
with another $\widetilde{A}:=(x,\widetilde{f}(x))$,
must follow certain common-sense rules.

First, any such replacement should not change the AMM's reserves:
if the AMM is in state $(a,f(a))$,
then the updated AMM is in state $(a,\widetilde{f}(a))$ where $f(a) = \widetilde{f}(a)$.
Adding or removing liquidity requires the active participation of the AMM's providers,
which can certainly happen,
but not as part of the kind of automatic strategy considered here.

Second,
any such replacement should not change the AMM's current exchange rate:
if the AMM is in state $(a,f(a))$,
then the updated AMM is in state $(a,\widetilde{f}(a))$ where $f'(a) = \widetilde{f}'(a)$.
To do otherwise invites further divergence loss.
If $(a,f(a))$ is the stable state for the current valuation,
and $f'(a) \neq \widetilde{f}'(a)$,
then $(a,\widetilde{f}(a))$ is not stable,
and a trader can make an arbitrage profit (and divergence loss)
by moving the AMM's state back to the stable state.

For example, an AMM's expected capitalization under distribution $p$ is 
\begin{equation*}
  \int_0^1 p(v) \bv \cdot \Phi(v) dv,
\end{equation*}
where $\Phi(v) = (\phi(v),f(\phi(v)))$.
If the distribution changes to $\widetilde{p}$,
then an adaptive strategy is to find a function $\widetilde{f}: \PosReals \to \PosReals$ with associated stable-point function
$\widetilde{\phi}: (0,1) \to \PosReals$ that optimizes
(or at least improves) the difference
\begin{equation*}
  \int_0^1 p(v) \bv \cdot \Phi(v) dv - \int_0^1 \widetilde{p}(v) \bv \cdot \widetilde{\Phi}(v)dv,
\end{equation*}
subject to boundary conditions $f(a) = \widetilde{f}(a)$ and $f'(a) = \widetilde{f'}(a)$,
where $\widetilde{\Phi}(v) = (\widetilde{\phi}(v),f(\widetilde{\phi}(v)))$.
Developing practical ways to find such functions is the subject of future work.

\section{Related Work}
\seclabel{related}
Today, the most popular automated market maker is
\emph{Uniswap}~\cite{uniswapv2,AngerisKCNC2019,uniswapv3,zhang2018},
a family of constant-product AMMs.
Originally trading between ERC-20 tokens and ether cryptocurrency,
later versions added direct trading between pairs of ERC-20 tokens,
and allowed liquidity providers
to restrict the range of prices in which their asset participate.
\emph{Bancor}~\cite{bancor} AMMs permit more flexible pricing schemes,
and later versions~\cite{bancorv2} include integration with external
``price oracles'' to keep prices in line with market conditions.
\emph{Balancer}~\cite{balancer} AMMs trade across more than two assets,
based on a \emph{constant mean} formula that generalizes constant product.
\emph{Curve}~\cite{curve} uses a custom curve specialized for trading
\emph{stablecoins} , maintaining low slippage and divergence loss
as long as the stablecoins trade at near-parity.
Pourpouneh \emph{et al.}~\cite{pourpouneh} provide a survey of current AMMs.

The formal model for AMMs used here,
including the axioms constraining AMM functions,
and notions of composition,
are taken from Engel and Herlihy~\cite{EngelH2021}.

Angeris and Chitra~\cite{AngerisC2020}
introduce a \emph{constant function market maker} model
and focus on conditions that ensure that agents who
interact with AMMs correctly report asset prices.

In \emph{event prediction markets}~\cite{AbernethyYV2011,ChenP2007,ChenW2010,Hanson2003,Hanson2007}, parties effective place bets on the outcomes of certain events, such as elections.
Event prediction AMMs differ from DeFi AMMs in important ways:
pricing models are different because
prediction outcome spaces are discrete rather than continuous,
prediction securities have finite lifetimes,
and composing AMMs is not a concern.

AMM curves resemble \emph{consumer utility curves} from classical economics ~\cite{micro},
and trader arbitrage resembles \emph{expenditure minimization}.
Despite some mathematical similarities,
there are fundamentally differences in application.
In particular,
traders interact with AMMs via composition,
an issue that does not arise in the consumer model.

Aoyagi~\cite{Aoyagi2020}
analyzes strategies for constant-product AMM liquidity providers
in the presence of ``noise'' trading, which is not intended to move prices,
and ``informed'' trading, intended to move the AMM to the stable point
for a new and more accurate valuation.

Angeris \emph{et al.}~\cite{AngerisEC2020}
propose an economic model relating how the curvature of the AMM's
function affects LP profitability in the presence of
informed and uninformed traders.

Bartoletti \emph{et al.}~\cite{BartolettiCL2021} give a formal semantics for a constant-product AMM expressed as a labeled transition system,
and formally verify a number of basic properties.

\bibliographystyle{plainurl}
\bibliography{references,zotero}

\newpage
\section{Appendix: Derivations of Equations}

\paragraph{\eqnref{linslipx}}
\begin{align*}
  \linslip_X(v,v'; A)
  &:= (1-v')\left(-\delta f'(x) + f(x + \delta) - f(x)\right) \\
  &= (1-v')\left(\delta \frac{v}{1-v} + f(x + \delta) - f(x)\right) \\
  &= \left(\frac{1-v'}{1-v}\right)\left(v\phi(v') - v\phi(v) + (1-v) f(\phi(v')) - (1-v) f(\phi(v))\right) \\
  &= \left(\frac{1-v'}{1-v}\right)\left( \bv \cdot \Phi(v') - \bv \cdot \Phi(v) \right) \\
  &= \left(\frac{1-v'}{1-v}\right) \divloss(v',v;A).
\end{align*}

\paragraph{\eqnref{angslip}}
\begin{align*}
    \angslip(v,v';A) 
    &= \theta(x') - \theta(x) \\
    &= \arctan\left (\frac{1}{-f'(x')} \right) - \arctan
    \left(\frac{1}{-f'(x)}\right) \\
    &= \arctan\left (\frac{1-v'}{v'}\right) - \arctan \left(\frac{1-v}{v} \right) \\
    &= \arctan\left (\frac{(\frac{1-v'}{v'}) - (\frac{1-v}{v})}{1 + (\frac{1-v'}{v'})(\frac{1-v}{v})}\right) \\
    &= \arctan\left(\frac{v-v'}{\bv \cdot \bv'}\right)
\end{align*}

\paragraph{\eqnref{divlossxy}}
\begin{align*}
    -\delta v_1' + v_2'(f(x) - f(x + \delta)) 
    &=v_1'(\phi_A(v_{12}) - \phi_A(v_{12}')) + v_2'(f(\phi_A(v_{12})) - f(\phi_A(v_{12}')) \\
    &=(v_1' + v_2')(\bv_{12}' \cdot \Phi_A(v_{12}) - \bv_{12}' \cdot \Phi_A(v_{12}')) \\
    &=(v_1' + v_2')\divloss(v_{12},v_{12}'; A) \\
    &= v_3'\divloss(v_{12},v_{12}'; A)
\end{align*}
\paragraph{\eqnref{divlossyz}}
\begin{align*}
    v_2'(f(x + \delta) - f(x)) &+ v_3'(g(y) - g(y + f(x) - f(x + \delta))) \\
    &=v_2'(\phi_B(v_{23}) - \phi_B(v_{23}')) + v_3'(g(\phi_B(v_{23})) - g(\phi_A(v_{23}'))\\
    &=(v_2' + v_3')(\bv_{23}' \cdot \Phi_B(v_{23}) - \bv_{23}' \cdot \Phi_B(v_{23}')) \\
    &=v_1' \divloss(v_{23},v_{23}'; B
\end{align*}

\paragraph{\eqnref{divlossxz}}
\begin{align*}
    -\delta v_1' + v_3'(h(x) - h(x + \delta))
    &=v_1'(\phi_{AB}(v_{13}) - \phi_{AB}(v_{13}')) +
    v_3'(h(\phi_{AB}(v_{13})) - h(\phi_{AB}(v_{13}'))\\
    &=(v_1' + v_3')((\bv_{13}' \cdot \Phi_{AB}(v_{13}) - \bv_{13}' \cdot \Phi_A(v_{13}')) \\
    &= v_2'\divloss(v_{13},v_{13}'; A \oplus B)
\end{align*}

\paragraph{\eqnref{divloss-seq}}
\begin{align*}
v_2' \divloss(v_{13}, v_{13}'; A \oplus B)
  &= v_3'\divloss(v_{12}, v_{12}'; A) + v_1' \divloss(v_{23},v_{23}'; B)\\
\divloss(v_{13}, v_{13}'; A \oplus B)
  &= \left( \frac{v_3'}{v_2'} \right) \divloss(v_{12}, v_{12}'; A) +
  \left( \frac{v_1'}{v_2'} \right) \divloss(v_{23}, v_{23}'; B).
\end{align*}

\paragraph{\eqnref{linslipxy}}
  \begin{align*}
  \linslip_X(v_{12},v_{12}'; A)
  &= \left(\frac{v_2'}{v_2}\right)
  \left( \frac{v_1+v_2}{v_1'+v_2'}\right)
  \left( \bv_{12} \cdot \Phi(v_{12}') - \bv_{12} \cdot \Phi(v_{12}) \right) \\
    v_2'(-\delta f'(x) + f(x + \delta) - f(x))
    &= v_2'\left(\frac{v_1}{v_2}(\phi_A(v_{12}') - \phi_A(v_{12})) + (f(\phi_A(v_{12}')) - f(\phi_A(v_{12}))\right) \\
    &= \frac{v_2'}{v_2}(v_1 + v_2)(\bv_{12} \cdot \Phi_A(v_{12}') - \bv_{12} \cdot \Phi_A(v_{12})) \\
    &= (v_1'+v_2')\linslip_X(v_{12}',v_{12}; A).
  \end{align*}
  
\paragraph{\eqnref{linslipyz}}
  \begin{align*}
    \linslip_X(v_{23},v_{23}'; B)
    &= \left(\frac{v_3'}{v_3}\right)
    \left( \frac{v_2+v_3}{v_2'+v_3'}\right)
    \left( \bv_{23} \cdot \Phi(v_{23}') - \bv_{23} \cdot \Phi(v_{23}) \right) \\
    &v_3' \left(
    (f(x)-(x+\delta))g'(y)
    + g(y+f(x)-(x+\delta))
    - g(y)
    \right)\\
    &=
    v_3' \left(
    (f(x)-(x+\delta))\frac{v_2}{v_3}
    + g(y+f(x)-(x+\delta))
    - g(y)
    \right)\\
    &=
    \frac{v_3'}{v_3}
    (v_2+v_3)
    \left(
    \bv_{23} \cdot \Phi_B(v_{23}')- \bv_{23} \cdot \Phi_B(v_{23}))
    \right)\\
    &=
    (v_2'+v_3') \linslip_X(v_{23},v_{23}'; B).
  \end{align*}
\paragraph{\eqnref{linslipxz}}
  \begin{align*}
    v_3'(-\delta h'(x) + h(x + \delta) - h(x))
    &= v_3'\left(\frac{v_1}{v_3}(\phi_A(v_{13}') - \phi_A(v_{13})) + (h(\phi_A(v_{13}')) - h(\phi_A(v_{13}))\right) \\
    &= \frac{v_3'}{v_3}(v_1 + v_3)(\bv_{13} \cdot \Phi_A(v_{13}') - \bv_{13} \cdot \Phi_A(v_{13})) \\
    &= (v_1'+v_3')\linslip_X(v_{13}',v_{13}, A \oplus B)
  \end{align*}

\paragraph{\eqnref{linslip-seq}}
  \begin{align*}
(v_1'+v_3')\linslip_X(v_{13},v_{13}'; A \oplus B)) &=
(v_1'+v_2')\linslip_X(v_{12},v_{12}' ; A) +
(v_2'+v_3')\linslip_X(v_{13},v_{13}' ; B)\\
\linslip_X(v_{13},v_{13}'; A \oplus B) &=
\left(\frac{v_1'+v_2'}{v_1'+v_3'}\right)\linslip_X(v_{12},v_{12}'; A) +
\left(\frac{v_2'+v_3'}{v_1'+v_3'}\right)\linslip_X(v_{23},v_{23}'; B)\\
&=
\left(\frac{1-v_3'}{1-v_2'}\right)\linslip_X(v_{12},v_{12}'; A) +
\left(\frac{1-v_1'}{1-v_2'}\right)\linslip_X(v_{23},v_{23}'; B)
  \end{align*}
  
\paragraph{\eqnref{linslip-seq}}
  \begin{align*}
  \angslip(v_{13},v_{13'}, A \oplus B)
  &= \arctan \left( \frac{v_3'}{v_1'} \right) - \arctan \left( \frac{v_3}{v_1} \right)\\
  &= \arctan \left(\frac{v_1 v_3'-v_1' v_3}{v_1 v_1'+v_3 v_3'}\right)\\
  &= \arctan \left(\frac{v_{13}-v_{13}'}{\bv_{13} \cdot \bv_{13}}\right)
  \end{align*}
  
\paragraph{\eqnref{load-seq}}
  \begin{align*}
  \load_X(v_{13},v_{13}'; A \oplus B)
  &=
  \divloss(v_{13},v_{13}'; A \oplus B) \linslip_X(v_{13},v_{13}'; A \oplus B) \\
  &=
  \left( \frac{v_3'(1-v_3')}{v_2'(1-v_2')} \right) \load_X(v_{12}, v_{12}'; A)
  +
  \left( \frac{v_1'(1-v_1')}{v_2'(1-v_2')} \right) \load_X(v_{23}, v_{23}'; B)\\
  &\qquad+
  \left( \frac{v_3'(1-v_3')}{v_2'(1-v_2')} \right)
  \divloss(v_{12}, v_{12}'; A)
  \linslip_X(v_{23},v_{23}'; B)\\
  &\qquad+
  \left( \frac{v_1'(1-v_1')}{v_2'(1-v_2')} \right)
  \divloss(v_{23}, v_{23}'; B)
  \linslip_X(v_{12},v_{12}'; A)\\
  \end{align*}

\paragraph{\eqnref{divloss-par}}
  \begin{align*}
\divloss(v,v';A||B)
    &=
    v'(-\delta) + (1-v')(f(x) - f(x + t\delta) + g(y) - g(y + (1-t)\delta)) \\
    &= -v'(x + t\delta - x + y + (1-t)\delta - y) \\
    &\qquad + (1-v')(f(\phi_A(v)) - f(\phi_A(v')) + g(\phi_B(v)) - g(\phi_B(v'))) \\ 
    &= -v'(\phi_A(v') - \phi_A(v) + \phi_B(v') - \phi_B(v))\\
    &\qquad + (1-v')(f(\phi_A(v)) - f(\phi_A(v')) + g(\phi_B(v)) - g(\phi_B(v'))) \\
    &= \bv' \cdot \phi_A(v) - \bv' \cdot \phi_A(v') + \bv' \cdot \phi_B(v) - \bv' \cdot \phi_B(v') \\
    &= \divloss(v,v';A) + \divloss(v,v';B)
  \end{align*}

\paragraph{\eqnref{linslip-par}}
  \begin{align*}
  h'(x) &= - (t f'(a+t x) + (1-t)g'(b+(1-t) x)\\
  &= - \left(t \frac{v}{(1-v)} + (1-t)\frac{v}{(1-v)} \right)\\
  &= - \frac{v}{(1-v)}  \\
    \delta &= t \delta + (1-t)\delta \\
    \delta &= \phi_A(v') - \phi_A(v) + \phi_B(v') - \phi_B(v) \\
    \linslip_X(v,v';A||B)
  &=
  -(1-v')( -\delta h'(x) + h(x+\delta) - h(x) )\notag\\
  &=
  -(1-v')\left( -\delta \frac{v}{1-v} + h(x+\delta) - h(x) \right)\notag\\
  &=
  \left(\frac{1-v'}{1-v}\right)
  (\delta v - (1-v)(h(x+\delta) - h(x)))\notag\\
  &=
  \left(\frac{1-v'}{1-v}\right)
  (\delta v - (1-v)(
  f(a) - f(a+t(x+\delta))+ g(b) - g(b+(1-t)(x+\delta))\notag\\
  &\qquad\qquad - (f(a) - f(a+t x)+ g(b) - g(b+(1-t)x))))\notag\\
  &=
  \left(\frac{1-v'}{1-v}\right)
  (\delta v + (1-v)(
  f(a+t(x+\delta)) - f(a+t x)\notag\\
  &\qquad\qquad+ g(b+(1-t)(x+\delta))- g(b+(1-t)x)))\notag\\
  &=
  \left(\frac{1-v'}{1-v}\right)
  (\delta v + (1-v)(
  f(\phi_A(v')) - f(\phi_A(v)) + g(\phi_B(v'))- g(\phi_B(v))))\notag\\
  &=
  \left(\frac{1-v'}{1-v}\right)
  (v(\phi_A(v') - \phi_A(v) + \phi_B(v') - \phi_B(v))\notag\\
  &\qquad\qquad+ (1-v)(f(\phi_A(v')) - f(\phi_A(v)) + g(\phi_B(v'))- g(\phi_B(v))))\notag\\
  &=
  \left(\frac{1-v'}{1-v}\right)
  (v(\phi_A(v') - \phi_A(v))
  + (1-v)(f(\phi_A(v')) - f(\phi_A(v)))\notag\\
  &\qquad\qquad+ v(\phi_B(v') - \phi_B(v)) + (1-v)(g(\phi_B(v'))- g(\phi_B(v))))\notag\\
  &=
  \left(\frac{1-v'}{1-v}\right)
  (\bv \cdot \Phi_A(v') - \bv \cdot \Phi_A(v'))
  + (\bv \cdot \Phi_B(v') - \bv \cdot \Phi_B(v'))\notag\\
  &=
  \linslip_X(v,v'; A) + \linslip_X(v,v'; B).
  \end{align*}
  
\paragraph{\eqnref{linslip-scalar}}
  \begin{align*}
    \linslip_X(v,v';\alpha A) 
    &= \bv' \cdot (\alpha x,\alpha f(x)) - \bv' \cdot (\alpha x',\alpha f(x')) \\
    &= \alpha(\bv' \cdot \Phi(v) - \bv' \cdot \Phi(v') \\
    &= \alpha \linslip_X(v,v';A).
  \end{align*}

\section{Appendix: Proofs}

\paragraph{\thmref{max-cap}}

\begin{proof}
    Let $h(v) = \bv \cdot \Phi(v) = v\phi(v) + (1-v)f(\phi(v))$ be the capitilization at $v$.
    Note that $h$ is a continuous function on the compact set $[0,1]$ which guarantees the existence of the maximum.
    Let $v^{*}$ be the point where the maximum occurs.
    The first derivative is
    \begin{align*}
        &h'(v) = \phi(v) + v\phi'(v) + (1-v)f'(\phi(v))\phi'(v) - f(\phi(v)) \\
        &\phi(v) + v\phi'(v) + (1-v)\frac{-v}{1 -v}\phi'(v) - f(\phi(v)) \\
        &= \phi(v) - f(\phi(v))
    \end{align*}
    The second derivative is
    \begin{align*}
        &h''(v) = \phi'(v) - f'(\phi(v))\phi'(v) \\
        &=\phi'(v)(1 - \frac{-v}{1-v}) \\
        &= \phi'(v)(\frac{1 - v + v}{1 -v}) = \frac{\phi'(v)}{1-v}
    \end{align*}
    
    Now take a derivative with respect to $v$ of
    \begin{align*}
        f'(\phi(v)) &= \frac{-v}{1-v}\\
        \phi'(v) &= \frac{-1}{(1-v)^2f''(\phi(v))} < 0.
    \end{align*}
    Recall that $f$ is strictly convex, so $f''(\phi(v)) > 0$ for all $v \in (0,1)$.
    
    We can then write the second derivative of the capitalization as
    \begin{align*}
        h''(v) = \frac{-1}{(1-v)^3f''(\phi(v))} < 0
    \end{align*}
    
    Thus $h$ is strictly concave so the maximum is unique.
    Finally, the first-order conditions tell us that $h'(v^{*}) = 0$ or $\phi(v^{*}) = f(\phi(v^{*}))$.
    
\end{proof}

\paragraph{\lemmaref{sym-curve}}
\begin{proof}
    Let $g = f^{-1} = f$.
    \begin{align*}
        f'(y) = g'(y) = \frac{1}{f'(f^{-1}(y))} 
        =\frac{1}{f'(f(y))}
        = \frac{1}{f'(x)} = - \frac{(1 -v)}{v}
        = \frac{-(1 -v)}{1 - (1-v)}
    \end{align*}
    Thus $y = \phi(1 -v)$.
\end{proof}

\paragraph{\thmref{sym-max-cap}}

\begin{proof}
    From the proof of \thmref{max-cap} we know $\phi'(v) < 0$ for $v \in (0,1)$, so $\phi(v)$ is strictly decreasing.
    Applying \thmref{max-cap} and \lemmaref{sym-curve} tells us that capitalization is maximized when $\phi(v) = x = y = \phi(1-v)$.
    Because $\phi$ is strictly decreasing, the only way $\phi(v) = \phi(1-v)$ is if $v = 1 -v$ or $v = \frac{1}{2}$.
\end{proof}

\section{Appendix: Minimizing Divergence Loss Exposure}
\seclabel{min-divcap}
Let $A:=(x,f(x))$ be an AMM currently in state $(a,f(a))$, the stable state for $(v,1-v)$.
Define an \emph{$X$-partition} to be a sequence of
values $\{x_i\}_{i=1}^\infty$ in $\PosReals$ such that 
$x_1 = a$, $x_i < x_{i+1}$, and $\lim_{n \to \infty}x_n = \infty$.

Given a partition $P_X = \{x_i\}_{i=1}^n$,
define the \emph{total divergence loss} with respect to that partition as
\begin{align*}
    \divloss(P_X;A) = \sum_{i=1}^\infty \divloss^*(x_i,x_{i+1};A)
\end{align*}

Writing this out explicitly gives
\begin{align*}
    &\sum_{i=1}^\infty \divloss^*(x_i,x_{i+1};A) \\
    &= \sum_{i=1}^\infty \bv_{i+1} \cdot (\Phi(\bv_i) - \Phi(\bv_{i+1)}) \\
    &= \sum_{i=1}^\infty \bv_{i+1} \cdot ((\phi(v_i),f(\phi(v_i))) - (\phi(v_{i+1}),f(\phi(v_{i+1})))) \\
    &= \sum_{i=1}^\infty (v_{i+1},1-v_{i+1}) \cdot (\phi(v_i) - \phi(v_{i+1}),f(\phi(v_i)) - f(\phi(v_{i+1}))) \\
    &= \sum_{i=1}^\infty v_{i+1}(\phi(v_i) - \phi(v_{i+1})) - v_{i+1}(f(\phi(v_i)) - f(\phi(v_{i+1}))) +\sum_{i=1}^\infty f(\phi(v_i)) - f(\phi(v_{i+1})) \\
    &= \sum_{i=1}^\infty v_{i+1}(\phi(v_i) - \phi(v_{i+1})) - v_{i+1}(f(\phi(v_i)) - f(\phi(v_{i+1}))) +f(\phi(v_1)) \\
    &= \sum_{i=1}^\infty v_{i+1}[\phi(v_i) - \phi(v_{i+1}) + f(\phi(v_{i+1})) - f(\phi(v_i))] +f(\phi(v_1)) \\
\end{align*}
Note that each term $\phi(v_i) - \phi(v_{i+1}) + f(\phi(v_{i+1})) - f(\phi(v_i)) \leq 0$.
We also know that $v_1 \geq v_{i+1}$ for each $i$.
This gives us the upper bound
\begin{equation*}
    \sum_{i=1}^\infty v_{i+1}[\phi(v_i) - \phi(v_{i+1}) + f(\phi(v_{i+1})) - f(\phi(v_i))] +f(\phi(v_1))
    \leq f(\phi(v_1)) = f(x)
\end{equation*}
and the lower bound
\begin{align*}
    &\sum_{i=1}^\infty v_{i+1}[\phi(v_i) - \phi(v_{i+1}) + f(\phi(v_{i+1})) - f(\phi(v_i))] +f(\phi(v_1)) \\
    &\geq v_1 \sum_{i=1}^\infty [\phi(v_i) - \phi(v_{i+1}) + f(\phi(v_{i+1})) - f(\phi(v_i))] +f(\phi(v_1)) \\
    &= v_1[\phi(v_1) - f(\phi(v_1))] + f(\phi(v_1)) \\
    &= v_1\phi(v_1) + (1-v_1)f(\phi(v_1)) \\
    &= \capital(v;A)
\end{align*}

Simply put
\begin{align*}
    \capital(v;A) \leq \divloss(P_X;A) \leq f(a)
\end{align*}

How tight can this lower bound get?
Well, let $v^{*}$ be the point where $\capital(v;A)$ is maximized.
If we let $x^{*} = \phi(v^{*})$, then we know that $f(x^{*}) = x^{*}$.
We know that $\capital(v;A) = v^{*}x^{*} + (1-v^{*})f(x^{*}) = f(x^{*})$.
But if $x = x^{*}$ then this means
\begin{align*}
        f(x^{*}) \leq \divloss(P;A) \leq f(x^{*})
\end{align*}
which is entirely independent of the chosen partition $P_X$, so
\begin{align*}
        f(x^{*}) \leq \divloss(A) \leq f(x^{*})
\end{align*}
or $\divloss(A) = f(x^{*})$.
Additionally, total loss is conserved even if we modify the AMM $A$ for $x > x^{*}$.
No matter how you choose to drain asset type $Y$ by depositing asset type $X$,
in the end you will drain all of $f(x^{*})$ if you deposit an infinite amount of $X$.

We get a similar result when trading along the $Y$-axis.
Define a \emph{$Y$-partition} to be a sequence of elements $\set{x_i}_{i=1}^\infty$ in $\PosReals$ such that $y_1 = f(a)$, $y_i \leq y_{i+1}$, and $\lim_{n \to \infty}y_n = \infty$.

Let $P_Y = \set{y_i}_{i=1}^{\infty}$ be a Y-partition, 
$g = f^{-1}(x)$.
and $y_i = f(x_i)$ so $x_i = f^{-1}(y_i) = g(y_i)$.
Thus $\Phi(\bv_i) = (x_i,f(x_i)) = (g(y_i),f \circ g(y_i))$ where $f \circ g(y_i) = f \circ f^{-1} (y_i) = y_i$.
That is, $\Phi(\bv_i) = (g(y_i),y_i)$.
The total cost with respect to this partition is
\begin{equation*}
    \divloss(P_Y;A) = \sum_{i=1}^\infty \divloss^*(y_i,y_{i+1};A)
\end{equation*}
For symmetry, define $v_i' = 1- v_i$.
\begin{align*}
    \sum_{i=1}^\infty \divloss^*(y_i,y_{i+1};A) \\
    &= \sum_{i=1}^\infty \bv_{i+1} \cdot (\Phi(\bv_i) - \Phi(\bv_{i+1)}) \\
    &= \sum_{i=1}^\infty \bv_{i+1} \cdot (g(y_i) - g(y_{i+1}),y_i - y_{i+1})  \\
    &=\sum_{i=1}^\infty (1-v_{i+1}')(g(y_i) - g(y_{i+1}) + v_{i+1}'(y_i - y_{i+1}) \\
    &=\sum_{i=1}^\infty v_{i+1}'[y_i - y_{i+1} +   g(y_{i+1}) - g(y_i)] + \sum_{i=1}^\infty (g(y_i) - g(y_{i+1}) \\
    &=\sum_{i=1}^\infty v_{i+1}'[y_i - y_{i+1} +   g(y_{i+1}) - g(y_i)] + a
\end{align*}
Note that $g(y_i) - g(y_{i+1})) + y_{i+1} - y_i \leq 0$ and $v_1' \geq v_{i+1}'$ for each $i$.
We now get the upper bound
\begin{align*}
    \sum_{i=1}^\infty v_{i+1}'[y_i - y_{i+1} +   g(y_{i+1}) - g(y_i)] + a \\
    &\leq a
\end{align*}
and the lower bound
\begin{align*}
    \sum_{i=1}^\infty v_{i+1}'[y_i - y_{i+1} +   g(y_{i+1}) - g(y_i)] + a \\
    &\geq v_1'\sum_{i=1}^\infty y_i - y_{i+1} +   g(y_{i+1}) - g(y_i) + a \\
    &= v_1'(y_1 - g(y_1)) + a \\
    &= (1-v_1)(f(a) - a) + a \\
    &= (1-v_1)f(a) + v_1a = \capital(v;A) 
\end{align*}
So again we get the bounds
\begin{align*}
    \capital(v;A) \leq \divloss(P_Y;A) \leq a
\end{align*}

Similar to the $X$-axis case this inequality is tight if $y_1 = f(x^{*})$ and $\divloss(P_Y;A) = a$.
Thus we do get a loss conservation result if we start at $(x^{*},f(x^{*}))$.
Namely
\begin{equation*}
    \divloss(P_X;A) + \divloss(P_Y;A) = x^{*} + f(x^{*}) = 2x^{*} = 2\capital(v^{*};A)  
\end{equation*}
The valuation $\bv^{*}$ thus corresponds to the AMM state where
half of the wealth may be lost to $X$ trades and half can be lost to $Y$ trades.

\section{Mathematica Code}
\seclabel{mathematica}
This section shows the Mathematica scripts used to generate \figref{expload}.
\begin{lstlisting}
g[x_] := InverseFunction[f][x];
\[Phi][v_] := InverseFunction[f'][-v/(1-v)];
divloss[v1_,v2_] :=
    v2 \[Phi][v1] + (1-v2) f[\[Phi][v1]] - (v2 \[Phi][v2] + (1-v2) f[\[Phi][v2]]);

linslipx[v1_,v2_] :=
    ((1-v2)/(1-v1)) (v1 \[Phi][v2] - v1 \[Phi][v1]
	+ (1-v1) f[\[Phi][v2]] - (1-v1) f[\[Phi][v1]]);

linslipy[v1_,v2_] :=
    (v2/v1) ((1-v1) f[\[Phi][v2]] - (1-v1) f[\[Phi][v1]]
	+ v1 \[Phi][v2] - v1 \[Phi][v1]);

loadx[v1_,v2_] := divload[v1,v2] linslipx[v1,v2];
loady[v1_,v2_] := divloss[v1,v2] linslipy[v1,v2];

exploadx[v_,\[Alpha]1_,\[Alpha]2_] :=
    PDF[BetaDistribution[\[Alpha]1, \[Alpha]2]][v] loadx[1/2,v];
exploady[v_,\[Alpha]1_,\[Alpha]2_] :=
    PDF[BetaDistribution[\[Alpha]1, \[Alpha]2]][v] loady[1/2,v];
expload[\[Alpha]1_, \[Alpha]2_] := 
 NIntegrate[exploadx[v, \[Alpha]1, \[Alpha]2], {v, 0, 1/2}]  + 
  NIntegrate[exploady[v, \[Alpha]1, \[Alpha]2], {v, 1/2, 1}]

expslipx[v_,\[Alpha]1_,\[Alpha]2_] :=
    PDF[BetaDistribution[\[Alpha]1, \[Alpha]2]][v] linslipx[1/2,v];
expslipy[v_,\[Alpha]1_,\[Alpha]2_] :=
    PDF[BetaDistribution[\[Alpha]1, \[Alpha]2]][v] linslipy[1/2,v];
expslip[\[Alpha]1_, \[Alpha]2_] := 
 NIntegrate[expslipx[v, \[Alpha]1, \[Alpha]2], {v, 0, 1/2}]  + 
  NIntegrate[expslipy[v, \[Alpha]1, \[Alpha]2], {v, 1/2, 1}]
  
exploadx[v_,\[Alpha]1_,\[Alpha]2_] :=
    PDF[BetaDistribution[\[Alpha]1, \[Alpha]2]][v] divloss[1/2,v];
exploady[v_,\[Alpha]1_,\[Alpha]2_] :=
    PDF[BetaDistribution[\[Alpha]1, \[Alpha]2]][v] divloss[1/2,v];
expload[\[Alpha]1_, \[Alpha]2_] := 
 NIntegrate[exploadx[v, \[Alpha]1, \[Alpha]2], {v, 0, 1/2}]  + 
  NIntegrate[exploady[v, \[Alpha]1, \[Alpha]2], {v, 1/2, 1}]

Plot3D[expload[\[Alpha]1, \[Alpha]2], {\[Alpha]1, 1, 4}, {\[Alpha]2, 
  1, 4}, PlotRange -> All, MeshFunctions -> {#3 &}]

\end{lstlisting}
\end{document}